\documentclass[conference]{IEEEtran}
\IEEEoverridecommandlockouts

\usepackage{amsmath}
\usepackage{amsfonts}
\usepackage{amssymb}
\usepackage{amsthm}
\usepackage{tabularx}
\usepackage{mathrsfs} 

\usepackage{url}
\usepackage{hyperref}
\newtheorem{theorem}{Theorem}

\newtheorem{corollary}{Corollary}

\usepackage{stfloats}
\usepackage{float}
\usepackage{graphicx}
\usepackage{cite}
\usepackage{xcolor}
\usepackage{subfigure}
\renewcommand{\vec}[1]{\mathbf{#1}}

\makeatletter
\def\blfootnote{\xdef\@thefnmark{}\@footnotetext}
\makeatother
\def\BibTeX{{\rm B\kern-.05em{\sc i\kern-.025em b}\kern-.08em
    T\kern-.1667em\lower.7ex\hbox{E}\kern-.125emX}}

\begin{document}

\title{On Performance of FAS-aided Wireless Powered \\NOMA Communication Systems
\thanks{This work was supported by the Engineering and Physical Sciences Research Council (EPSRC) under Grant EP/W026813/1.}
}

\author{
\IEEEauthorblockN{1\textsuperscript{st} Farshad Rostami Ghadi}
\IEEEauthorblockA{\textit{Dept. of Electronic and Electrical Engineering} \\
\textit{University College London}\\
London, United Kingdom \\
f.rostamighadi@ucl.ac.uk}
\and
\IEEEauthorblockN{2\textsuperscript{nd} Masoud Kaveh}
\IEEEauthorblockA{\textit{Dept. of Information and Communication Engineering} \\
\textit{Aalto University}\\
Espoo, Finland \\
masoud.kaveh@aalto.fi}
\and
\IEEEauthorblockN{3\textsuperscript{rd} Kai-Kit Wong}
\IEEEauthorblockA{\textit{Dept. of Electronic and Electrical Engineering} \\
\textit{University College London}\\
London, United Kingdom \\
kai-kit.wong@ucl.ac.uk}
\and
\IEEEauthorblockN{4\textsuperscript{th} Riku Jäntti}
\IEEEauthorblockA{\textit{Dept. of Electrical Engineering} \\
\textit{Aalto University}\\
Espoo, Finland \\
riku.jantti@aalto.fi}
\and
\IEEEauthorblockN{5\textsuperscript{th} Zheng Yan}
\IEEEauthorblockA{\textit{School of Cyber Engineering} \\
\textit{Xidian University}\\
Xi'an, China \\
zyan@xidian.edu.cn}
\and
}

	\maketitle
 
	\begin{abstract}
This paper studies the performance of a wireless powered communication network (WPCN) under the non-orthogonal multiple access (NOMA) scheme, where users take advantage of an emerging fluid antenna system (FAS). More precisely, we consider a scenario where a transmitter is powered by a remote power beacon (PB) to send information to the planar NOMA FAS-equipped users through Rayleigh fading channels. After introducing the distribution of the equivalent channel coefficients to the users, we derive compact analytical expressions for the outage probability (OP) in order to evaluate the system performance. Additionally, we present asymptotic OP in the high signal-to-noise ratio (SNR) regime. Eventually, results reveal that deploying the FAS with only one activated port in NOMA users can significantly enhance the WPCN performance compared with using traditional antenna systems (TAS). 
	\end{abstract}
 
	\maketitle
	
	\vspace{0mm}
	\section{Introduction}\label{sec-intro}
	\IEEEPARstart{O}{ne} of the most paramount challenges in designing the next generation of mobile communication systems, i.e., sixth-generation (6G) technology, is the limited battery life of mobile devices. Indeed, battery replacement or recharging can incur significant costs due to the extensive deployment of mobile devices, which severely limits the energy efficiency of wireless communication networks \cite{bi2015wireless}. To tackle this issue, thanks to the recent advance in radio frequency (RF) signals-based wireless energy transfer (WET) technology, wireless powered communication network (WPCN) has been introduced as a promising paradigm to remotely power up mobile devices by providing continuous microwave energy over the air \cite{bi2016wireless}. 
 More precisely, in contrast to conventional approaches,  WPCN is capable of configuring the transmit power, waveforms, occupied time and frequency dimensions, and other critical network parameters based on physical conditions and service requirements. Moreover, this technology is able to efficiently transmit tens of microwatts of RF power to mobile devices located at short-range distances, which makes it potentially suitable for low-power devices in various applications such as the Internet of Things (IoT), device-to-device (D2D) communications, RF identification (RFID) networks, and  backscatter communications (BC).
	
	Apart from the aforementioned challenge, another important obstacle in future mobile technology is scaling the number of antennas in the small space of mobile devices to enhance multiplexing gains and network capacity, which affects the system performance in terms of power consumption, implementation complexity, channel estimation, signal processing, etc. In this regard, the fluid antenna system (FAS) has recently emerged as a potential technology to enhance the multiple-input multiple-output (MIMO) system, introducing a new degree of freedom through antenna position flexibility \cite{wong2020fluid}. Specifically, a fluid antenna is a software-manageable fluidic or dielectric structure that can change its shape and position to configure the radiation characteristics. This important property, which is the main superiority of FAS over traditional antenna systems (TAS), is due to recent developments in utilizing liquid metals and RF switchable pixels for antennas, which can potentially enhance the overall system performance and transmission reliability  \cite{wong2022bruce}.
	
Over the recent years, great efforts have been carried out to evaluate the effect of WPCN in different applications such as IoT \cite{chen2019resource,vu2021performance,do2020enabling}, D2D \cite{deng2018energy,van2020performance,lu2017wireless}, and BC \cite{han2017wirelessly,ghadi2022capacity,lyu2017wireless}. Moreover, outstanding contributions have been recently made related to deploying FAS in various wireless communication scenarios, mainly focusing on channel modeling \cite{wong2021fluid,khammassi2023new,ghadi2023copula,ghadi2023gaussian} and estimation \cite{skouroumounis2022fluid,xu2023channel}, performance analysis \cite{new2023fluid,ghadi2023fluid1,tlebaldiyeva2022enhancing,ghadi2024performance,vega2023novel,ghadi2024cache}, and optimization \cite{chen2023energy,wang2024fluid,xu2023capacity,ye2023fluid}. Furthermore, it is worth noting that most existing works in the context of FAS are based on the orthogonal multiple access (OMA) technique, which is not an optimum scheme when the channel state information (CSI) at the transmitter is known. For this purpose, only recently, several works considered the non-OMA (NOMA) principle for the FAS, which is more efficient compared with OMA, thanks to the superposition coding and successive interference cancellation (SIC) method \cite{new2023fluid1,tlebaldiyeva2023full,zheng2024fas}.
	
Nevertheless, the synergy of combining FAS with WPCN is not thoroughly grasped. The only research in this context has been carried out in \cite{ghadi2024performance1}, where the authors considered a one-dimensional (1D) FAS-equipped tag in the BC (i.e., the product channel) under the OMA scheme. In contrast to the previous work, in this paper, we propose a more general WPCN, where a single-antenna transmitter, which is remotely powered up by a power beacon (PB), simultaneously sends information to NOMA FAS-equipped users. Besides, each NOMA user includes a planar that is able to switch to the best position (i.e., port) in a pre-defined two-dimensional (2D) space for reception. In this regard, the main contributions of this work are as follows: (i) By assuming that each channel undergoes Rayleigh fading, we first derive the cumulative distribution function (CDF) of the equivalent channel at NOMA FAS-equipped users with the help of the Gaussian copula; (ii) Then, we derive compact analytical expressions for the outage probability (OP) in terms of the multivariate normal CDF; (iii) Further, we obtain asymptotic expressions for the OP to gain insights regarding system's performance in high signal-to-noise ratio (SNR) regime; (iv) Finally, after conducting extensive simulations, the results indicate that applying FAS in WPCN is beneficial for NOMA users, ensuring more reliable transmission compared with deploying TAS. 
	\section{System Model}\label{sec-sys}
	\subsection{Chanel Model}
We consider a WPCN as shown in Fig. \ref{fig-system}, where a transmitter $\mathrm{t}$ simultaneously broadcasts the superposed signal including information to the NOMA users $i\in\left\{\mathrm{u}_1,\mathrm{u}_2\right\}$ over wireless fading channels, exploiting the energy wirelessly supplied by a PB. In this regard, the wireless channel between the PB and the transmitter $\mathrm{t}$ is called the \textit{energy link}, while the channels between the transmitter $\mathrm{t}$ and near user $\mathrm{u}_1$ (i.e., strong user) and between the transmitter $\mathrm{t}$ and far user $\mathrm{u}_2$ (i.e., weak user) are referred to as \textit{information links}. It is assumed that the PB and transmitter $\mathrm{t}$ are single-antenna nodes, while NOMA users $i$ are equipped with a planar FAS. More precisely, the FAS-equipped NOMA users $i$ include a grid structure with $N^l_i$ ports that are  uniformly distributed along a linear space of length $W^l_i\lambda$ for $l\in\left\{1,2\right\}$, i.e., $N_i=N^1_i\times N^2_i$ and $W_i=W^1_i\lambda\times W^2_i\lambda$. Moreover, in order to transform the 2D indices to the 1D index, an applicable mapping function $\mathcal{F}\left(n_i\right)=\left(n_i^1,n_i^2\right)$, $n_i=\mathcal{F}^{-1}\left(n_i^1,n_i^2\right)$ is defined, where $n_i\in\left\{1,\dots,N_i\right\}$ and $n_i^l\in\left\{1,\dots,N_i^l\right\}$. Under such assumptions, the received signal at the $n_i$-th port of the user $i$ is expressed as
		\begin{align}
		y_i^{n_i}=\sqrt{P_\mathrm{p}\delta_i}h_{\mathrm{eq},i}^{n_i}\left(\sqrt{p_\mathrm{u_1}}x_\mathrm{u_1}+\sqrt{p_\mathrm{u_2}}x_\mathrm{u_2}\right)+z_i^{n_i},
	\end{align}
in which $P_\mathrm{p}$ denotes the transmitted power by the PB and $h_{\mathrm{eq},i}^{n_i}=h_\mathrm{t}h_i^{n_i}$ is the equivalent channel coefficient at the $n_i$-th port of the user $i$ that include the channel coefficients of the energy link $h_\mathrm{t}$ and the information links $h_i^{n_i}$. Besides, $x_i$ represents the symbol sent to the user $i$ with unit power, $p_i$ defines the power allocation factor for user $i$, where $p_\mathrm{u_1}+p_\mathrm{u_2}=1$, and $z_i^{n_i}$ is the additive white Gaussian noise (AWGN) with zero mean and variance $\sigma^2$ at the $n_i$-th port of user $i$. Moreover, $\delta_i=L_\mathrm{p}d_\mathrm{t}^{-\alpha}d_i^{-\alpha}$ defines the path-loss, in which $L_\mathrm{p}$ is the frequency-dependent signal propagation losses, $\alpha>2$ denotes the path-loss exponent, and $d_\mathrm{t}$ is the distance between the PB and transmitter $\mathrm{t}$ and $d_i$ represents the distance between transmitter $\mathrm{t}$ and the NOMA FAS-equipped user $i$. 
	
Further, the ports in FAS can freely switch and be in close proximity to each other; thereby, the channel coefficients at each NOMA FAS-equipped user exhibit spatial correlation. By assuming that the energy and information links suffer Rayleigh fading, the covariance between two arbitrary ports $n_i$ and $\tilde{n}_i$ at each user $i$ in a three-dimensional (3D) environment under rich scattering can be defined as 
\vspace{-6pt}
	\begin{align}
	\varrho_{n_i,\tilde{n}_i}^i=j_0\left(2\pi\sqrt{\left(\frac{n^1_i-\tilde{n}^1_i}{N_i^1-1}W_i^1\right)^2+\left(\frac{n^2_i-\tilde{n}^2_i}{N_i^2-1}W_i^2\right)^2}\right),
	\end{align}
	 in which $\tilde{n}_i=\mathcal{F}^{-1}\left(\tilde{n}_i^1,\tilde{n}_i^2\right)$ so that $\tilde{n}_i^l\in\left\{1,\dots,N_i^l\right\}$ and $j_0(.)$ denotes  the zero-order spherical Bessel function of the first kind. 
	 \begin{figure}[!t]
	 	\centering
	 \includegraphics[width=0.9\columnwidth]{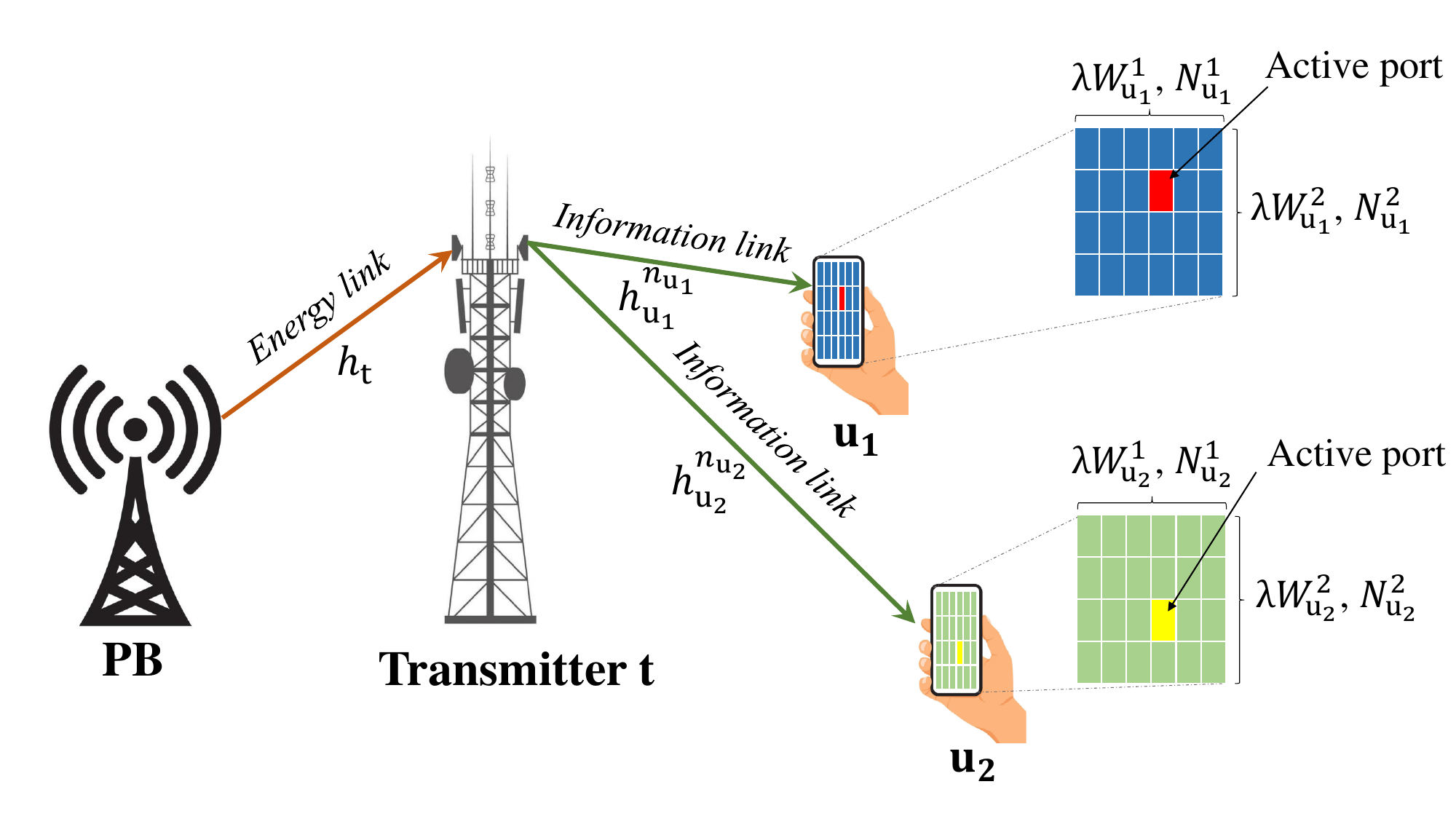}\vspace{-0.5cm}
	 	\caption{System model: FAS-aided WPCN.}\vspace{-0.5cm}
	 	\label{fig-system}
	 \end{figure}
  \vspace{-4pt}
	 \subsection{NOMA-FAS Scheme}
Following the NOMA principle, the optimal decoding technique involves implementing SIC according to the strong FAS-equipped user $\mathrm{u}_1$, which means that the weak FAS-equipped user $\mathrm{u}_2$ must decode its signal directly by treating interference as noise. Hence, by assuming that only the optimal port that maximizes the received signal-to-interference-noise ratio (SINR) at the FAS-equipped users is activated, the SINR of the SIC process can be given by
\vspace{-6pt}
	\begin{align}
	\gamma_\mathrm{sic}=\frac{\overline{\gamma}p_\mathrm{u_2}L_\mathrm{p}d_\mathrm{t}^{-\alpha}d_\mathrm{u_1}^{-\alpha}\left|h_{\mathrm{eq},{\mathrm{u_1}}}^{n^*_{\mathrm{u_1}}}\right|^2}{\overline{\gamma}p_\mathrm{u_1}L_\mathrm{p}d_\mathrm{t}^{-\alpha}d_\mathrm{u_1}^{-\alpha}\left|h_{\mathrm{eq},{\mathrm{u_1}}}^{n^*_{\mathrm{u_1}}}\right|^2+1},\label{eq-sinr-sic}
	\end{align}	 
	where $\overline{\gamma}=\frac{P_\mathrm{p}}{\sigma^2}$ defines the average SNR and $n_i^*$ denotes the best port's index at the user $i$, i.e., 
 
	\begin{align}
		n_i^*=\underset{n}{\arg\max}\left\{\left|h_{\mathrm{eq},i}^{n_i}\right|^2\right\}.
	\end{align} 
	Thus, the channel gain at the FAS-equipped user $i$ is given by
		\begin{align}
		g_{\mathrm{fas},i}=\max\left\{\left|h_{\mathrm{eq},{i}}^1\right|^2,\dots,\left|h_{\mathrm{eq},{i}}^{N_i}\right|^2\right\}, \label{eq-g-fas}
	\end{align}
	in which 
	$g_{\mathrm{eq},i}^{n_i}=\left|h_{\mathrm{eq},{i}}^{n_i}\right|^2$. 
	
After SIC, $\mathrm{u}_1$ removes the message of $\mathrm{u}_2$ from its received signal and then decodes its required information. Therefore, the received SNR at the FAS-equipped user $\mathrm{u}_1$ is defined as 
		\begin{align}
		\gamma_\mathrm{u_1}=\overline{\gamma}p_\mathrm{u_1}L_\mathrm{p}d_\mathrm{t}^{-\alpha}d_\mathrm{u_1}^{-\alpha}\left|h_{\mathrm{eq},{\mathrm{u_1}}}^{n^*_{\mathrm{u_1}}}\right|^2. \label{eq-snr-u1}
	\end{align}
At the same time, the weak user $\mathrm{u}_2$ directly decodes its own signal but lacks the capability to filter out the signal from the strong user $\mathrm{u}_2$ within the combined transmitted message. Therefore, the received SINR at the weak FAS-equipped user $\mathrm{u}_2$ is expressed as
			\begin{align}
		\gamma_\mathrm{u_2}=\frac{\overline{\gamma}p_\mathrm{u_2}L_\mathrm{p}d_\mathrm{t}^{-\alpha}d_\mathrm{u_2}^{-\alpha}\left|h_{\mathrm{eq},{\mathrm{u_2}}}^{n^*_{\mathrm{u_2}}}\right|^2}{\overline{\gamma}p_\mathrm{u_1}L_\mathrm{p}d_\mathrm{t}^{-\alpha}d_\mathrm{u_2}^{-\alpha}\left|h_{\mathrm{eq},{\mathrm{u_2}}}^{n^*_{\mathrm{u_2}}}\right|^2+1}.
	\end{align}
\section{Performance Analysis}
Here, after introducing the distribution of the equivalent channel to the FAS-equipped users, we obtain analytical expressions of the OP for both NOMA users. 
\subsection{Statistical Characteristics}
In order to mathematically analyze the important performance metrics in the considered FAS-aided WPCN, first, it is required to derive the distribution of the equivalent fading channel coefficients at the NOMA users, i.e., the maximum of $N_i$ correlated RVs so that each includes the product of two independent exponentially-distributed RVs (see \eqref{eq-g-fas}). In this regard, given that energy and information links undergo Rayleigh fading, their corresponding channel gains  $g_\mathrm{t}=\left|h_\mathrm{t}\right|^2$ and $g_i^{n_i}=\left|h_i^{n_i}\right|^2$ have the exponential distribution with unit mean. Therefore, the CDF of the equivalent channel gain $g_{\mathrm{eq},i}^{n_i}=g_\mathrm{t}g_i^{n_i}$ at the user $i$ is determined as
	\begin{align}
	F_{g_{\mathrm{eq},i}^{n_i}}(r)&=\int_0^\infty f_{g_\mathrm{t}}(g_\mathrm{t})F_{g_i^{n_i}}\left(\frac{r}{g_\mathrm{t}}\right)\mathrm{d}g_\mathrm{t}\\
	&=1-\int_0^\infty \exp\left(-\left(g_\mathrm{t}+\frac{r}{g_\mathrm{t}}\right)\right)\mathrm{d}g_\mathrm{t}\\
	& \overset{(a)}{=}1-2\sqrt{r}\mathcal{K}_1\left(2\sqrt{r}\right), \label{eq-cdf-eq}
	\end{align}
	in which $f_{g_\mathrm{t}}(\cdot)$ and $F_{g_i^{n_i}}(\cdot)$ are the marginal probability density function (PDF) and CDF of exponentially-distributed $g_\mathrm{t}$ and $g_i^{n_i}$. Besides, $(a)$ is obtained by utilizing the integral format provided in \cite[3.471.9]{gradshteyn2007table} in which $\mathcal{K}_1(.)$ represents the first-order modified Bessel function
	of the second kind. 
	
	Next, by using the definition, the CDF of $g_{\mathrm{fas},i}$ can be mathematically expressed as follows
 
	\begin{align}
		F_{\mathrm{fas},i}(r)&=\Pr\left(\max\left\{g_{\mathrm{eq},i}^1,\dots,g_{\mathrm{eq},i}^{N_i}\right\}\leq r\right)\\
		&=\Pr\left(g_{\mathrm{eq},i}^1\leq r,\dots,g_{\mathrm{eq},i}^{N_i}\leq r\right)\\
		&=F_{g_{\mathrm{eq},i}^1,\dots,g_{\mathrm{eq},i}^{N_i}}\left(r,\dots,r\right), \label{eq-fas1}
	\end{align}
where \eqref{eq-fas1} refers to the definition of the multivariate distribution of $N_i$ correlated random variables (RVs) $g_{\mathrm{eq}_i}^{n_i}$ for $n_i\in\left\{1,\dots,N_i\right\}$. To derive such a distribution, we exploit Sklar's theorem \cite{ghadi2023gaussian} which can generate the multivariate distribution of $N_i$ correlated arbitrary RVs, e.g., $g_{\mathrm{fas},i}$, by only knowing the corresponding marginal distributions. Hence, $F_{\mathrm{fas},i}(r)$ can be given by
	\begin{align}
		F_{g_{\mathrm{fas},i}}\left(r\right)=\Phi_{\vec{R}_i}\left(\phi^{-1}\left(F_{g_{\mathrm{eq},i}^1}\left(r\right)\right)\hspace{-1mm},\dots,\phi^{-1}\left(F_{g_{\mathrm{eq},i}^{N_i}}\left(r\right)\right)\hspace{-1mm};\vartheta_i\right),\label{eq-cdf-fas}
	\end{align}
	where $\phi^{-1}\left(F_{g_{\mathrm{eq},i}^{n_i}}\left(r\right)\right)=\sqrt 2\mathrm{erf}^{-1}\left(2F_{g_{\mathrm{eq},i}^{n_i}}\left(r\right)-1\right)$ indicates the quantile function of the standard normal distribution and $\vartheta_i$ is the dependence parameter, i.e., the Gaussian copula parameter \footnote{The Gaussian copula can accurately model the spatial correlation between fluid antenna ports, especially when the fluid antenna size is large. In other words, the dependence parameter of the Gaussian copula can be approximated in terms of the covariance of two arbitrary ports, i.e., $\varrho^i\approxeq\vartheta_i$ \cite{ghadi2023gaussian}.}. The term $\Phi_{\vec{R}_i}$ denotes the joint CDF of the multivariate normal distribution with a zero mean vector and the following correlation matrix $\vec{R}_i$
	\begin{align}
		\vec{R}_i=\begin{bmatrix}
			\varrho_{1,1}^i & \varrho_{1,2}^i &\dots& \varrho_{1,N_i}^i\\
			\varrho_{2,1}^i & \varrho_{2,2}^i &\dots& \varrho_{2,N_i}^i\\ \vdots & \vdots & \ddots & \vdots\\
			\varrho_{N_i,1}^i & \varrho_{N_i,2}^i &\dots& \varrho_{N_i,N_i}^i
		\end{bmatrix}.
	\end{align}
	\subsection{OP Analysis}
	OP is a key performance metric to assess wireless communication systems, which is defined as the probability that the instantaneous SNR $\gamma$ is below the SNR threshold $\gamma_\mathrm{th}$, i.e., $P_\mathrm{out}=\Pr\left(\gamma\leq\gamma_\mathrm{th}\right)$.
	\begin{theorem}\label{thm-op-u1}
		The OP for the strong FAS-equipped user $\mathrm{u}_1$ over the considered FAS-aided WPCN is given by \eqref{eq-op1}, in which 
		\begin{figure*}
			\begin{align}
				P_\mathrm{out,u_1} = \Phi_{\vec{R}_{\mathrm{u_1}}}\left(\sqrt{2}\mathrm{erf}^{-1}\left(1-4\sqrt{\tilde{\gamma}_\mathrm{max}}\mathcal{K}_1\left(2\sqrt{\tilde{\gamma}_\mathrm{max}}\right)\right),\dots,\sqrt{2}\mathrm{erf}^{-1}\left(1-4\sqrt{\tilde{\gamma}_\mathrm{max}}\mathcal{K}_1\left(2\sqrt{\tilde{\gamma}_\mathrm{max}}\right)\right)\hspace{0mm};\vartheta_1\right)\label{eq-op1}
			\end{align}
		\hrulefill
		\end{figure*}
		\begin{align}
		\tilde{\gamma}_\mathrm{sic}=\frac{\hat{\gamma}_\mathrm{sic}d_\mathrm{t}^{\alpha}d_\mathrm{u_1}^{\alpha}}{\overline{\gamma}L_\mathrm{p}\left(p_\mathrm{u_2}-\hat{\gamma}_\mathrm{sic}p_\mathrm{u_1}\right)}, \quad \tilde{\gamma}_\mathrm{u_1}=\frac{\hat{\gamma}_\mathrm{u_1}d_\mathrm{t}^{\alpha}d_\mathrm{u_1}^{\alpha}}{\overline{\gamma}p_\mathrm{u_1}L_\mathrm{p}}.\label{eq-tildes}
		\end{align}
	\end{theorem} 
	\begin{proof}
		OP arises for the strong user $\mathrm{u}_1$ when it fails to decode its own signal, the weak user's signal, or both. Thus, by definition, the OP for $\mathrm{u_1}$ can be mathematically expressed as
		\begin{align}
			P_\mathrm{out,u_1} = 1-\Pr\left(\gamma_\mathrm{sic}>\hat{\gamma}_{\mathrm{sic}},\gamma_\mathrm{u_1}>\hat{\gamma}_\mathrm{u_1}\right), \label{eq-def-op1}
		\end{align}
		in which $\hat{\gamma}_\mathrm{sic}$ is the SINR threshold of $\gamma_{\mathrm{sic}}$ and $\hat{\gamma}_\mathrm{u_1}$ denotes the SNR threshold of $\gamma_\mathrm{u_1}$. Given that the OP in NOMA scenario is mainly affected by the power allocation factors, $P_\mathrm{out,u_1}$ for the case that $p_\mathrm{u_2}\leq\hat{\gamma}_\mathrm{sic}p_\mathrm{u_1}$ becomes $1$. Therefore, for the case that $p_\mathrm{u_2}>\hat{\gamma}_\mathrm{sic}p_\mathrm{u_1}$, the OP can be computed as follows
		\begin{align}\notag
	&	P_\mathrm{out,u_1} \overset{(b)}{=} 1-\Pr\Bigg(\frac{\overline{\gamma}p_\mathrm{u_2}L_\mathrm{p}d_\mathrm{t}^{-\alpha}d_\mathrm{u_1}^{-\alpha}g_{\mathrm{fas},\mathrm{u_1}}}{\overline{\gamma}p_\mathrm{u_1}L_\mathrm{p}d_\mathrm{t}^{-\alpha}d_\mathrm{u_1}^{-\alpha}g_{\mathrm{fas},\mathrm{u_1}}+1}>\hat{\gamma}_{\mathrm{sic}},\\
	&\hspace{3cm}\overline{\gamma}p_\mathrm{u_1}L_\mathrm{p}d_\mathrm{t}^{-\alpha}d_\mathrm{u_1}^{-\alpha}g_{\mathrm{fas},\mathrm{u_1}}>\hat{\gamma}_\mathrm{u_1}\Bigg)\\\nonumber
	&=1-\Pr\Bigg(g_{\mathrm{fas},\mathrm{u_1}}>\frac{\hat{\gamma}_\mathrm{sic}d_\mathrm{t}^{\alpha}d_\mathrm{u_1}^{\alpha}}{\overline{\gamma}L_\mathrm{p}\left(p_\mathrm{u_2}-\hat{\gamma}_\mathrm{sic}p_\mathrm{u_1}\right)},\\
	&\hspace{3cm}g_{\mathrm{fas},\mathrm{u_1}}>\frac{\hat{\gamma}_\mathrm{u_1}d_\mathrm{t}^{\alpha}d_\mathrm{u_1}^{\alpha}}{\overline{\gamma}p_\mathrm{u_1}L_\mathrm{p}}\Bigg)\\
	& = 1-\Pr\left(g_\mathrm{fas,u_1}>\max\left\{\tilde{\gamma}_\mathrm{sic},\tilde{\gamma}_\mathrm{u_1}\right\}\right)\\
	& = F_{g_{\mathrm{fas},\mathrm{u_1}}}\left(\tilde{\gamma}_\mathrm{max}\right), \label{eq-pr-u1}
	\end{align}
	in which $(b)$ is followed by inserting \eqref{eq-sinr-sic} and \eqref{eq-snr-u1} into \eqref{eq-def-op1}, 		$\tilde{\gamma}_\mathrm{sic}=\frac{\hat{\gamma}_\mathrm{sic}d_\mathrm{t}^{\alpha}d_\mathrm{u_1}^{\alpha}}{\overline{\gamma}L_\mathrm{p}\left(p_\mathrm{u_2}-\hat{\gamma}_\mathrm{sic}p_\mathrm{u_1}\right)}$, $ \tilde{\gamma}_\mathrm{u_1}=\frac{\hat{\gamma}_\mathrm{u_1}d_\mathrm{t}^{\alpha}d_\mathrm{u_1}^{\alpha}}{\overline{\gamma}p_\mathrm{u_1}L_\mathrm{p}}$, and $\tilde{\gamma}_\mathrm{\max}=\max\left\{\tilde{\gamma}_\mathrm{sic},\tilde{\gamma}_\mathrm{u_1}\right\}$. Then, by considering the CDF of $g_{\mathrm{fas},i}$ from $\eqref{eq-cdf-fas}$, the proof is accomplished.  
	\end{proof}
	\begin{theorem}
			The OP for the weak FAS-equipped user $\mathrm{u}_2$ over the considered FAS-aided WPCN is given by \eqref{eq-op2}, in which 
					\begin{figure*}
				\begin{align}
					P_\mathrm{out,u_2} =\, \Phi_{\vec{R}_{\mathrm{u_2}}}\left(\sqrt{2}\mathrm{erf}^{-1}\left(1-4\sqrt{\tilde{\gamma}_\mathrm{u_2}}\mathcal{K}_1\left(2\sqrt{\tilde{\gamma}_\mathrm{u_2}}\right)\right),\dots,\sqrt{2}\mathrm{erf}^{-1}\left(1-4\sqrt{\tilde{\gamma}_\mathrm{u_2}}\mathcal{K}_1\left(2\sqrt{\tilde{\gamma}_\mathrm{u_2}}\right)\right)\hspace{0mm};\vartheta_2\right)\label{eq-op2}
				\end{align}
				\hrulefill
			\end{figure*}
						\begin{align}
				\tilde{\gamma}_\mathrm{u_2}=\frac{\hat{\gamma}_\mathrm{u_2}d_\mathrm{t}^{\alpha}d_\mathrm{u_2}^{\alpha}}{\overline{\gamma}L_\mathrm{p}\left(p_\mathrm{u_2}-\hat{\gamma}_\mathrm{u_2}p_\mathrm{u_1}\right)}. \label{eq-tilde2}
			\end{align}
	\end{theorem}

\begin{proof}
	By definition, the OP for the weak user $\mathrm{u}_2$ can be defined as
	\begin{align}
		P_\mathrm{out,r_2} = \Pr\left(\gamma_\mathrm{u_2}\leq\hat{\gamma}_\mathrm{u_2}\right),
	\end{align}
	in which $\hat{\gamma}_{\mathrm{u_2}}$ is the SINR threshold of SINR $\gamma_{\mathrm{u_2}}$. Following the same approach in the proof of Theorem \ref{thm-op-u1}, the OP for $\mathrm{u}_2$ when $p_\mathrm{u_2}>\hat{\gamma}_\mathrm{u_2}p_\mathrm{u_1}$ can be derived as
	\begin{align}
	&P_\mathrm{out,u_2} = \Pr\left(\frac{\overline{\gamma}p_\mathrm{u_2}L_\mathrm{p}d_\mathrm{t}^{-\alpha}d_\mathrm{u_2}^{-\alpha}g_\mathrm{fas,u_2}}{\overline{\gamma}p_\mathrm{u_1}L_\mathrm{p}d_\mathrm{t}^{-\alpha}d_\mathrm{u_2}^{-\alpha}g_\mathrm{fas,r_2}+1}\leq\hat{\gamma}_\mathrm{u_2}\right)\\
	&=\Pr\left(g_{\mathrm{fas},\mathrm{u_2}}\leq\frac{\hat{\gamma}_\mathrm{u_2}d_\mathrm{t}^{\alpha}d_\mathrm{u_2}^{\alpha}}{\overline{\gamma}L_\mathrm{p}\left(p_\mathrm{u_2}-\hat{\gamma}_\mathrm{u_2}p_\mathrm{u_1}\right)}\right)\\
	&=F_{g_{\mathrm{fas},\mathrm{u_2}}}\left(\frac{\hat{\gamma}_\mathrm{u_2}d_\mathrm{t}^{\alpha}d_\mathrm{u_2}^{\alpha}}{\overline{\gamma}L_\mathrm{p}\left(p_\mathrm{u_2}-\hat{\gamma}_\mathrm{u_2}p_\mathrm{u_1}\right)}\right). \label{eq-pr-u2}
	\end{align}
Now, by considering the CDF of $g_{\mathrm{fas},i}$, \eqref{eq-op2} is obtained and the proof is completed. 
	\end{proof}
Despite OP expressions in \eqref{eq-op1} and \eqref{eq-op2} accurately assessing the proposed system performance, we are interested in analyzing the OP behavior in the high SNR regime to gain more insights. For this purpose, by utilizing the series expansion of $\mathcal{K}_1\left(r\right)$ when $r\rightarrow 0$, the CDF of the equivalent channel gain in \eqref{eq-cdf-eq} can be written as \cite{ghadi2024performance1}
\begin{align}
F_{g_{\mathrm{eq},i}^{n_i}}(r)\approx r\left(1-2\zeta-\ln r\right), \label{eq-cdf-asy}
\end{align}
in which $\zeta$ is the Euler-Mascheroni constant. 
\begin{corollary}
The OP for the string and weak FAS-equipped users over the considered FAS-aided WPCN is given by \eqref{eq-op1-asy} and \eqref{eq-op2-asy}, respectively, where $\tilde{\gamma}_\mathrm{sic}$ and  $\tilde{\gamma}_\mathrm{u_1}$ is defined in \eqref{eq-tildes}, and $\tilde{\gamma}_\mathrm{sic}$ is provided in \eqref{eq-tilde2}.
\begin{proof}
By applying \eqref{eq-cdf-asy} to \eqref{eq-pr-u1} and \eqref{eq-pr-u2} the proof is completed.  \vspace{-0.5cm}
\end{proof}
	\begin{figure*}
	\begin{align}
		P_\mathrm{out,u_1}^\infty = \Phi_{\vec{R}_{\mathrm{u_2}}}\left(\sqrt{2}\mathrm{erf}^{-1}\left(2 \tilde{\gamma}_\mathrm{max}\left(1-2\zeta-\ln \tilde{\gamma}_\mathrm{max}\right)-1\right),\dots,\sqrt{2}\mathrm{erf}^{-1}\left(2\tilde{\gamma}_\mathrm{max}\left(1-2\zeta-\ln \tilde{\gamma}_\mathrm{max}\right)-1\right)\hspace{0mm};\vartheta_1\right)\label{eq-op1-asy}
	\end{align}\vspace{-0.5cm}
	\hrulefill
\end{figure*}
\begin{figure*}
	\begin{align}
		P_\mathrm{out,u_2}^\infty =\, \Phi_{\vec{R}_{\mathrm{u_2}}}\left(\sqrt{2}\mathrm{erf}^{-1}	\left(2 \tilde{\gamma}_\mathrm{u_2}\left(1-2\zeta-\ln \tilde{\gamma}_\mathrm{u_2}\right)-1\right),\dots,\sqrt{2}\mathrm{erf}^{-1}	\left(2 \tilde{\gamma}_\mathrm{u_2}\left(1-2\zeta-\ln \tilde{\gamma}_\mathrm{u_2}\right)-1\right)\hspace{0mm};\vartheta_2\right)\label{eq-op2-asy}
	\end{align}
	\hrulefill
\end{figure*}
\end{corollary}
\section{Numerical Results}\label{sec-num}
In this section, we assess the analytical derivations of the OP, which are double-checked with the Monte Carlo simulation. To do so, we set the channel model parameters as $p_\mathrm{u_1}=0.3$, $p_\mathrm{u_2}=0.7$, $\alpha=2.5$, $L_\mathrm{p}=1$,  $d_\mathrm{t}=d_\mathrm{u_2}=10$m, $d_\mathrm{u_1}=5$m, and  $\hat{\gamma}_{\mathrm{u_1}}=\hat{\gamma}_{\mathrm{u_2}}=\hat{\gamma}_{\mathrm{{sic}}}=0$dB. Furthermore, given that the OP derivations are in terms of the joint multivariate normal distribution, we implement them numerically through the mathematical package of the programming language MATLAB. Moreover, we utilize the proposed algorithm in \cite{ghadi2023gaussian} to simulate the Gaussian copula in the considered WPCN.
\begin{figure*}[]
	\begin{center}
		\vspace{-0.5cm}
		\subfigure[]{\includegraphics[width=.67\columnwidth]{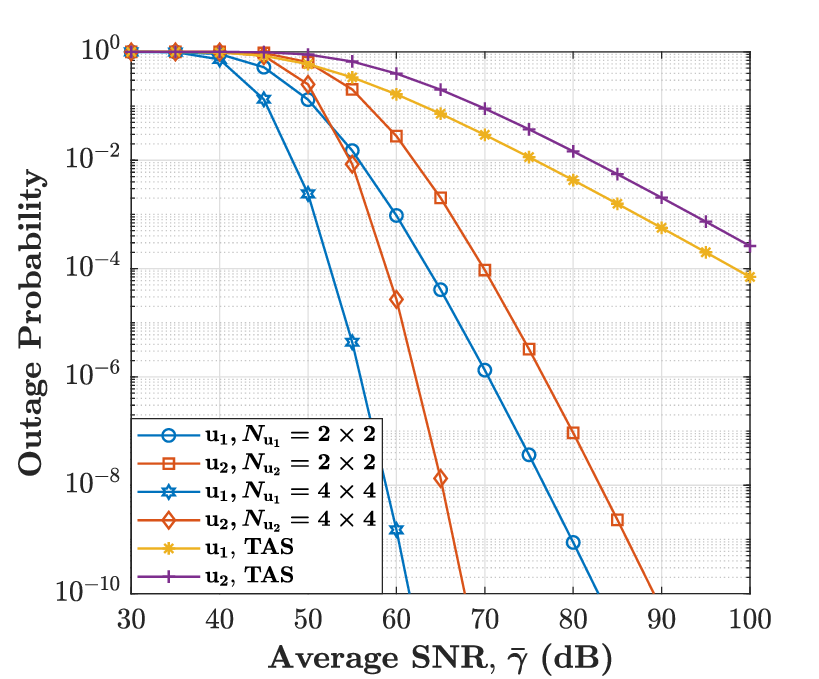}\label{fig-1}}\hspace{2.5cm}
		\subfigure[]{\includegraphics[width=.67\columnwidth]{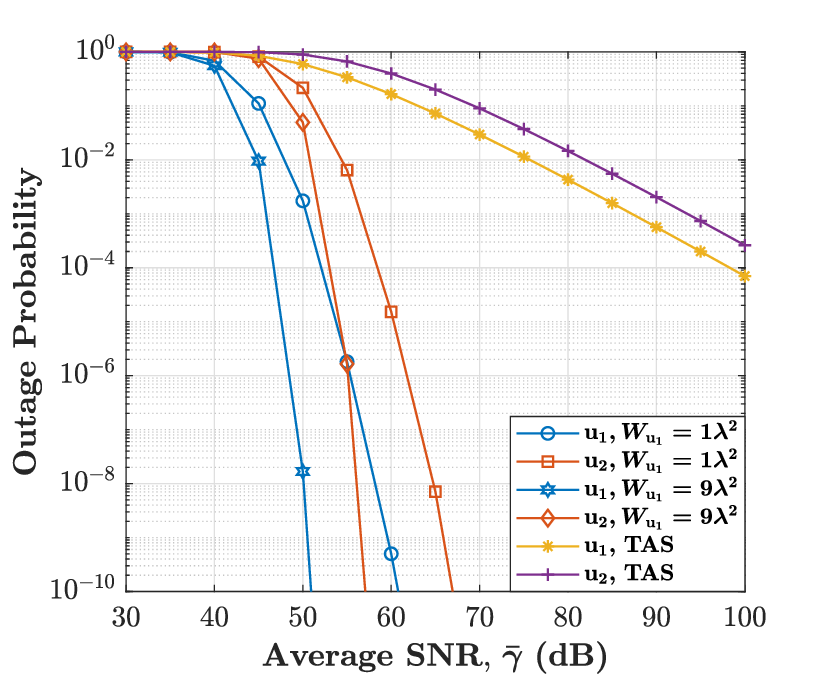}\label{fig-2}}\vspace{-0.3cm}
		\caption{OP versus  $\overline{\gamma}$ for the strong and weak NOMA FAS-equipped users $\mathrm{u_1}$ and $\mathrm{u_2}$ (a) when $W_\mathrm{u_1}=W_\mathrm{u_2}=1 \lambda^2$; (b) when $N_\mathrm{u_1}=N_\mathrm{u_2}=25$.}\label{figs1}\vspace{0cm}
	\end{center}
\end{figure*}
\begin{figure*}[]
	\begin{center}
		\vspace{-0.6cm}
		\subfigure[]{\includegraphics[width=.67\columnwidth]{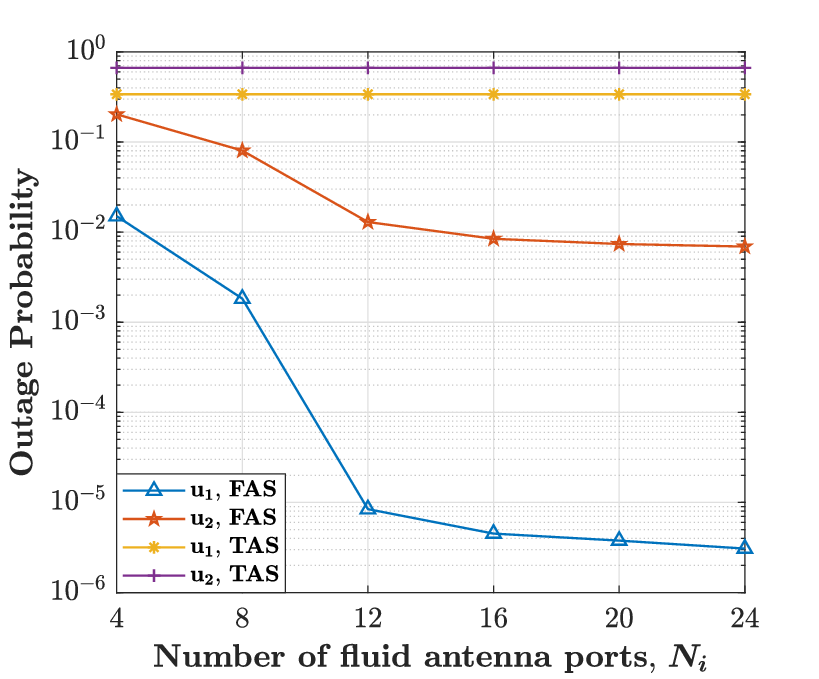}\label{fig-3}}\hspace{2.5cm}
		\subfigure[]{\includegraphics[width=.67\columnwidth]{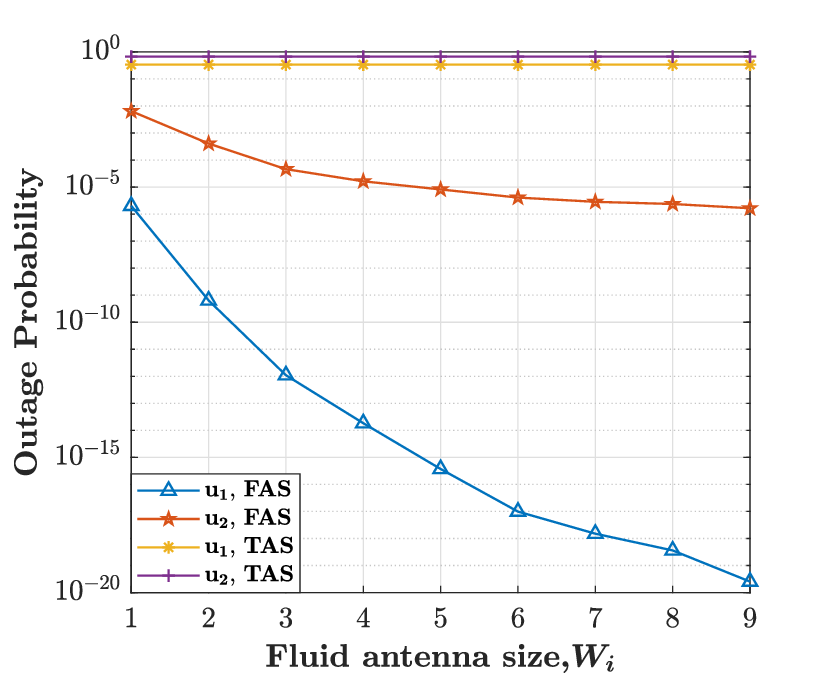}\label{fig-4}}\vspace{-0.3cm}
		\caption{(a) OP versus $N_i$ for the strong and weak NOMA FAS-equipped users $\mathrm{u_1}$ and $\mathrm{u_2}$ when $W_\mathrm{u_1}=W_\mathrm{u_2}=1\lambda^2$ and $\overline{\gamma}=55$dB; (b) OP versus $W_i$ for the strong and weak NOMA FAS-equipped users $\mathrm{u_1}$ and $\mathrm{u_2}$ when $N_\mathrm{u_1}=N_\mathrm{u_2}=25$ and $\overline{\gamma}=55$dB. }\label{figs2}\vspace{-0.7cm}
	\end{center}
\end{figure*}

Figs. \ref{fig-1} and \ref{fig-2} illustrate the OP performance for both strong and weak NOMA users in terms of the average SNR $\overline{\gamma}$ when different numbers of fluid antenna ports $N_i$ and various values of fluid antenna size $W_i$ are considered. In both figures, it can be seen that as $\overline{\gamma}$ grows, the OP for both users decreases, which is obvious since the channel quality improves. However, we can observe that increasing $\overline{\gamma}$ provides a lower OP for the strong user $\mathrm{u}_1$ compared with the weak user $u_2$. The key reason behind this trend is that $\mathrm{u}_1$ takes advantage of the SIC which provides a larger SINR in contrast to $\mathrm{u}_2$. Moreover, it is evident that increasing the number of fluid antenna ports or the fluid antenna size can significantly enhance the OP performance for both strong and weak users, so that this behavior becomes more tangible in the higher SNR regime. More precisely, the main reason for this behavior is that although increasing $N_i$ under a fixed value of $W_i$ raises the spatial correlation between fluid antenna ports, it is also able to simultaneously enhance channel capacity, diversity gain, and spatial multiplexing. Besides, by increasing $W_i$ for a constant number of $N_i$, the spatial separation between fluid antenna ports decreases, which leads to a weaker spatial correlation. In this regard, we can also see that NOMA FAS-equipped users with only one activated port experience more reliable transmission compared with the TAS-equipped users, i.e., the single fixed-antenna system.  For instance, under a fixed value of $\overline{\gamma}=60$dB, the OP for the strong FAS-equipped user $\mathrm{u}_1$ with $N_\mathrm{u_1}=4$  and $W_\mathrm{u_1}=1\lambda^2$ is around $10^{-3}$, while the OP for a strong user with a single fixed-antenna is close to $10^{-1}$. 

In order to gain more insights into how the OP performance for NOMA users changes based on the number of fluid antenna ports and the fluid antenna size, we present Figs. \ref{figs2}. From Fig. \ref{fig-3}, it can be seen that the OP performance for both users initially enhances and then becomes saturated as $N_i$ continuously grows. This is mainly due to the fact that increasing $N_i$ under a fixed $W_i$ leads to the ports becoming too close to each other and the spatial correlation effects dominating; thereby, diversity gain reduces after reaching a certain point and the decrease in OP slows down so that eventually saturates. Additionally, we can observe that when $N_i$ increases, the reduction in OP for the strong NOMA user $\mathrm{u}_1$ tends to be much more significant than for the weak NOMA user $\mathrm{u}_2$. The key reason behind such a behavior is that $\mathrm{u}_1$ being less susceptible to interference due to its higher SINR, can take better advantage of the interference mitigation technique, i.e., SIC, resulting in a larger reduction in the OP. Fig. \ref{fig-4} also indicates that as $W_i$ grows mainly due to a fixed large number of fluid antenna ports, e.g., $N_i=25$, the OP reduction is unhindered, albeit decreasing slowly at the larger $W_i$. This is mainly because increasing $W_i$ can significantly reduce the spatial correlation between fluid antenna ports, resulting in improved signal reception, reduced interference, and ultimately, a lower OP. Moreover, we can observe that increasing $W_i$ achieves a lower OP for $\mathrm{u}_1$ compared with $\mathrm{u}_2$ since the strong user benefits more from the interference mitigation and enhanced spatial diversity and multiplexing gain. For instance, increasing $W_i$ from $1\lambda^2$ to $9\lambda^2$ under a fixed $N_i=25$ for $i\in\{\mathrm{u}_1, \mathrm{u_2}\}$ changes the OP of $\mathrm{u}_1$ from the order of $10^{-6}$ to around $10^{-20}$, whereas the OP for $\mathrm{u}_2$  decreases from the order of $10^{-2}$ to  the order of $10^{-6}$. In contrast, it can be seen that the OP for both NOMA users under single fixed-antenna deployment is almost near $1$ since the diversity gain remains constant even if $W_i$ increases.\vspace{-0.01cm}

\section{Conclusion}\label{sec-con}
In this paper, we considered a WPCN under the NOMA scheme, where a single-antenna transmitter, which is powered up by a remote PB, simultaneously sends the superposed signal including messages to NOMA users. Besides, we assumed that each NOMA user takes advantage of a planar FAS that is capable of changing its position in a pre-set 2D space to reach the maximum SNR. In order to evaluate the system performance, after introducing the distribution of the equivalent channel to FAS-equipped users, we derived the compact theoretical expressions for the OP in terms of the multivariate normal CDF. Further, we obtained asymptomatic OP in the high SNR regime. Results showed that FAS-equipped users experience more reliable transmission than TAS-equipped users in the NOMA WPCN.\vspace{-0.05cm}

\bibliographystyle{IEEEtran}
\bibliography{sample.bib}

\end{document}